\documentclass{amsart}

\usepackage{amsmath,amssymb}
\usepackage{amsthm}
\usepackage{latexsym}
\usepackage{graphicx}
\usepackage{color}
\usepackage{subfigure}
\usepackage{mathptmx}

\newcommand{\NN}{\mathbb{N}}

\newcommand{\ZZ}{\mathbb{Z}}

\newcommand{\RR}{\mathbb{R}}
\newcommand{\CC}{\mathbb{C}}

\theoremstyle{definition}
\newtheorem{theorem}{Theorem}[section]
\newtheorem{defn}[theorem]{Definition}

\newtheorem{thm}{Theorem}[section]

\title{When interpolation-induced reflection artifact meets time-frequency analysis}

\author[Yu-Ting~Lin]{Yu-Ting~Lin}
\address{Department of Anesthesiology, Shin Kong Wu Ho-Su Memorial Hospital, Taipei, Taiwan} 
\email{seagaia.lin@msa.hinet.net}
\author[Patrick~Flandrin]{Patrick~Flandrin}
\address{Physics Department (UMR 5672 CNRS), Ecole Normale Sup\'erieure de Lyon, 46 all\'ee dItalie, 69364 Lyon Cedex 07 France}
\email{Patrick.Flandrin@ens-lyon.fr}
\author[Hau-Tieng Wu]{Hau-Tieng Wu}
\address{Department of Mathematics, University of Toronto, Toronto Ontario Canada}
\email{hauwu@math.toronto.edu}

\begin{document}

\maketitle

\begin{abstract}
While extracting the temporal dynamical features based on the time-frequency analyses, like the reassignment and synchrosqueezing transform, attracts more and more interest in bio-medical data analysis, we should be careful about artifacts generated by interpolation schemes, in particular when the sampling rate is not significantly higher than the frequency of the oscillatory component we are interested in.
In this study, we formulate the problem called the reflection effect and provide a theoretical justification of the statement. We also show examples in the anesthetic depth analysis with clear but undesirable artifacts. 
The results show that the artifact associated with the reflection effect exists not only theoretically but practically. Its influence is pronounced when we apply the time-frequency analyses to extract the time-varying dynamics hidden inside the signal. 
In conclusion, we have to carefully deal with the artifact associated with the reflection effect by choosing a proper interpolation scheme. 

\end{abstract}

\section{Introduction}
It has been widely accepted that several aspects of the health status could be well observed by analyzing recorded physiological time series. In particular, the {\em time-varying oscillatory pattern} inside the electrocardiogram (ECG) or respiratory signal contains abundant health information, for example, the heart rate variability (HRV) \cite{Malik_Camm:1995,TaskForce:1996,Lewis_Furman_McCool_Porges:2012} hidden inside the R peak to R peak interval (RRI) time series and the instantaneous heart rate (IHR), the breathing pattern variability (BPV) representing the time varying rate of the respiratory signal \cite{Benchetrit:2000,Wysocki_Cracco_Teixeira_Mercat_Diehl_Lefort_Derenne_Similowski:2006,Wu_Hseu_Bien_Kou_Daubechies:2013}. It is well known that power spectrum is not a suitable tool when the time-varying dynamics in the signal is the main target to analyze, as power spectrum reflects only the {\em global} oscillatory information, and hence could not properly extract the dynamical information, which is {\em local} in nature. In general, a popular and powerful way to study the time-varying oscillatory pattern inside a time series is the time-frequency (TF) analysis, which allows us to efficiently extract how a signal oscillates at each time instant. There have been several TF analysis techniques proposed, including linear methods like short time Fourier transform (STFT), continuous wavelet transform (CWT) \cite{Daubechies:1992,Flandrin:1999} and multiwindow approach \cite{Jaillet_Torresani:2007}, the quadratic methods like Wigner-Ville distribution and Cohen class \cite{Flandrin:1999}, nonlinear methods like reassignment (RM) technique \cite{Auger_Flandrin:1995,Chassande-Mottin_Auger_Flandrin:2003}, synchrosqueezing transform (SST) \cite{Daubechies_Lu_Wu:2011,Chen_Cheng_Wu:2014}, multi-tapered RM \cite{Xiao_Flandrin:2007}, multi-tapered SST \cite{Lin_Wu_Tsao_Yien_Hseu:2014,Lin:2015Thesis}, ConceFT \cite{Daubechies_Wang_Wu:2015}, empirical model decomposition \cite{Huang_Shen_Long_Wu_Shih_Zheng_Yen_Tung_Liu:1998}, sparse TF analysis \cite{Hou_Shi:2011}, iterative filtering \cite{Cicone_Liu_Zhou:2014}, etc. The potential of the TF analysis has been shown in several different fields, in particular the medical field; for example, \cite{Huang_Chan_Lin_Wu_Huang:1997,SouzaNeto_Loiseau_Cejka_Custaud_Abry_Frutoso_Gharib_Flandrin:2007,Orini_Bailon_Mainardi_Laguna_Flandrin:2012,Lin_Wu_Tsao_Yien_Hseu:2014}, to name but a few.

While there exists a lot of information in the physiological signals we could easily approach, in some signals, like IHR, there are two particular features that should not be neglected. First, they are sampled in a non-uniform fashion; second, in many situations, they are often sampled at a rate which is not significantly high. 
Typical examples include the IHR and the ECG-derived respiratory (EDR) signal extracted from the ECG signal, the IHR estimated from the photoplethysmography signal \cite{Gil_Orini_Bailon_Vergara_Mainardi_Laguna:2010}, etc, in which case the sampling is non-uniform and the sampling rate is determined by the heart rate. 
In order to apply the TF analysis to study these observed signals, a common practice is to apply the digital-to-analogue conversion to recover the original continuous signal \cite{Moody_Mark_Zocoola_Mantero:1985,Malik_Camm:1995,Huang_Chan_Lin_Wu_Huang:1997,Chui_Lin_Wu:2014,Lin_Wu_Tsao_Yien_Hseu:2014}, for example, the spline interpolation. Among different choices of the spline interpolation schemes, the cubic spline interpolation is commonly chosen as it balances between the interpolation property and the overfitting issue raised by the interpolation. On the other hand, the interpolation is commonly associated with a convolutional kernel. While it is generally accepted that the quality of recovering the underlying signal is good by a suitably chosen interpolation scheme, it might cause some undesirable artifacts, for example, by the side lobe effect inherited in the spline interpolation scheme \cite{Unser:1999}.
Thus, although the TF analysis have been successfully applied to several physiological problems for different purposes, we should be careful about the analysis results. For example, while the main purpose of the nonlinear TF analysis techniques like RM and SST is to sharpen/enhance the TF representation determined by the linear TF analysis, these techniques might also enhance the artifact generated by the interpolation scheme and hence mislead the interpretation. This problem is further complicated by the possible non-uniform sampling scheme. 

To demonstrate the potential problem caused by this issue, here we show a real example in the respiratory signal analysis. The airflow respiratory signal and the ECG signal are recorded simultaneously, and we obtain the EDR signal from the ECG signal. The EDR signal is generated by interpolating the amplitudes of detected R peaks via the cubic spline interpolation, and is influenced by the inevitable sampling effect inherited in the R peak location, which is not only non-uniform but also with low sampling rate. In Figure \ref{FigExample}, the respiratory signal and the EDR signal are shown together, as well as their TF representations determined by the multi-tapered RM. Clearly, while these two signals ``look'' similar in the sense of the ``fast-slow'' oscillatory pattern, their TF representations are very different. While the TF representation of the respiratory signals show the multiples of the base respiratory frequency at about 0.5Hz, the TF representation of the EDR signal show ``two different components'', where the component with the higher frequency is not the multiple of the base respiratory frequency at about 0.5Hz. Clearly, the TF representation of the EDR signal will mislead the interpretation.

\begin{figure}[t]
\centering
\includegraphics[width=0.9\textwidth]{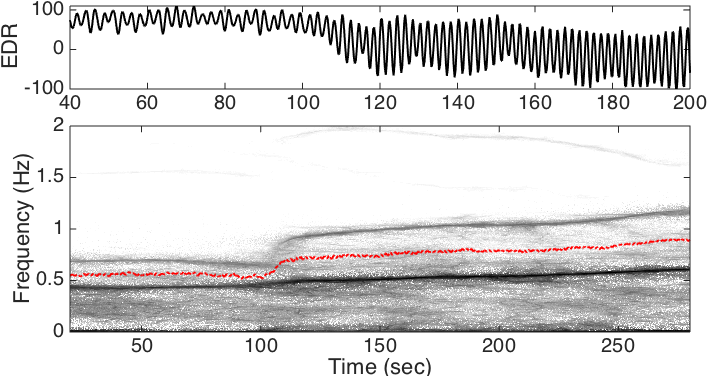}
\includegraphics[width=0.9\textwidth]{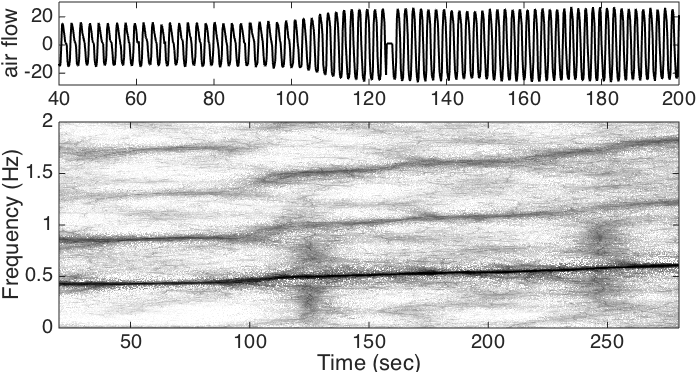}
\caption{The time-frequency representations of the ECG derived respiratory (EDR) signal and the airflow signal determined by the multi-taper reassignment (RM). The signals are recorded simultaneously. Top: the EDR signal based on the cubic spline interpolation from the R peaks amplitudes determined from the lead II ECG signal. The mean of the EDR signal is removed. Middle top: the multi-taper RM of the EDR. The instantaneous Nyquist frequency (INF) is superimposed as a red dashed curve. Middle bottom: the airflow signal recorded simultaneously. At around 150 second, an automatic calibration happens, which leads to a short zero period. Bottom: the multi-taper RM of the airflow signal. Note that the breathing rate is below 0.5 Hz in the beginning and increases gradually. It is clear that the base component with instantaneous frequency (IF) about $0.5$ Hz in the airflow signal is well captured by the EDR signal, while in the EDR signal there is an artificial component with IF higher than INF, which mirrors the base component IF via the INF.}\label{FigExample} 
\end{figure}

Thus, to extract the correct information from this kind of time series, it is essential to understand the influence on the analysis results caused by the sampling scheme, the interpolation scheme and the TF analysis. In this paper, the reflection effect caused by the interpolation scheme is formalized, how the sampling scheme is involved is discussed, and a theoretical justification of this effect is provided. In addition, we show several medical examples where the artifacts might mask the whole interpretation.

This paper is organized in the following way. In Section \ref{Section:ModelNyquist}, the adaptive harmonic model is introduced to model the commonly encountered oscillatory signals. The notion instantaneous Nyquist rate is introduced to quantify the nonuniform sampling that serves as a framework for our analysis. In Section \ref{Section:reflection}, the reflection effect is formalized and a theoretical justification is provided. In addition, a solution to this artificial effect is proposed. A series of numerical evidences, as well as real examples from the anesthesia are provided in Sections \ref{Section:Numerical} and \ref{Section:Real}. In Section \ref{Section:Conclusion}, we conclude the paper with a series of discussions.

\section{The adaptive harmonic model and instantaneous Nyquist rate}\label{Section:ModelNyquist}

{Among different features of time series, oscillatory pattern is the main target of several TF analysis techniques. When it comes to the oscillation, Fourier analysis is commonly the first choice among others. However, in general, the oscillatory pattern might change according to time, and its momentary behavior might not be easily captured by the Fourier technique and new techniques are needed. While the time-varying oscillatory pattern might be as complicated as one could imagine, in medical applications it is commonly controlled by the physiological constraints. For example, the heart rate could not be as fast as possible, and the amplitude of the respiratory signal is limited by the lung capacity. Thus, to quantify the oscillatory pattern inside an observed time series, we consider the following {\em adaptive harmonic model} \cite{Daubechies_Lu_Wu:2011,Chen_Cheng_Wu:2014}:}
\begin{defn}
Fix $0< \epsilon\ll 1$, $\epsilon\ll c_1< c_2$. The functional class, $\mathcal{A}^{c_1,c_2}_{\epsilon}\subset C^1(\RR)\cap L^\infty(\RR)$, contains functions of the form $a(t)\cos(2\pi\phi(t))$, where $a\in C^1(\RR)\cap L^\infty(\RR)$, $\phi\in C^2(\RR)$, $c_1\leq a(t)\leq c_2$, $c_1\leq \phi'(t)\leq c_2$ $|a'(t)|\leq \epsilon \phi'(t)$ and $|\phi''(t)|\leq \epsilon \phi'(t)$ for all time $t\in\RR$. 
\end{defn}
For $f(t)=a(t)\cos(2\pi\phi(t))\in \mathcal{A}^{c_1,c_2}_{\epsilon}$, we call $\phi'(t)$ the instantaneous frequency (IF) and $a(t)$ the amplitude modulation (AM) of $f(t)$. In the adaptive harmonic model, locally all functions in $\mathcal{A}^{c_1,c_2}_{\epsilon}$ behave like a harmonic function (or sine wave), and the deviation from being a harmonic function is controlled by $\epsilon$. We call a function in $\mathcal{A}^{c_1,c_2}_{\epsilon}$ an {\it intrinsic mode type} (IMT) function.
A more general model consisting of functions with multiple IMT functions and more theoretical discussions could be found in \cite{Daubechies_Lu_Wu:2011,Chen_Cheng_Wu:2014}; the identifiability issue of the adaptive harmonic model and the noise and trend model issue have been studied and discussed in \cite{Chen_Cheng_Wu:2014}. While the following discussion carries over into to the more general model, in this paper, to simplify the discussion, we focus on $\mathcal{A}^{c_1,c_2}_{\epsilon}$, where there is only one IMF inside the signal. 
The adaptive harmonic model has been applied to study the ``non-stationary'' physiological dynamics \cite{Wu_Hseu_Bien_Kou_Daubechies:2013,Lin_Wu_Tsao_Yien_Hseu:2014,Baudin_Wu_Bordessoule_Beck_Jouvet_Frasch_Emeriaud:2014,Wu_Talmon_Lo:2014,Orini_Bailon_Mainardi_Laguna_Flandrin:2012,Lin:2015Thesis}. 
Before proceeding, we comment that most time series we could acquire in the real world are real, so we consider only real model but not the exponential complex model. We mention that obtaining the imaginary part of a given real signal to recover its exponential complex form is not a trivial problem \cite{Bedrosian:1962,Nuttall:1966,Huang_Wang_Yang:2013}.

Next we discuss the sampling effect. Take $f(t)=a(t)\cos(2\pi\phi(t))\in \mathcal{A}^{c_1,c_2}_{\epsilon}$. Consider a monotonic increasing function $\psi\in C^1(\RR)$. 
Define a sequence of sampling points $\{t_m\}_{m\in\ZZ}$ so that $t_m=\psi^{-1}(m)$. With this sampling scheme, we obtain samples $\mathcal{X}:=\{t_m, f(t_m)\}_{m\in\ZZ}$. Note that if $\psi(t)=kt$, where $k>0$, then the sampling is uniform and we get a sample every $1/k$ second. 
To study how much information about $f$ we could obtain from the sampled dataset, when $\psi$ is a linear function and $f$ is a band-limited function, we could consider the Nyquist-Shannon theory. However, the application of Nyquist-Shannon theory is not efficient in our case, since a generic function in $\mathcal{A}^{c_1,c_2}_{\epsilon}$ has a non-compact support\footnote{For the definition of support, we refer the reader to \cite[P. 284]{Folland:1999}.} in the Fourier domain. The application of the Nyquist-Shannon theory is more difficult in the case when $\psi$ is nonlinear; that is, when the sampling is non-uniform. Thus, we consider another definition to describe the sampling scheme which reflects the momentary nature of dynamical analysis.

\begin{defn}\label{Definition:ISR}
Fix $c>0$ and $0\leq\epsilon\ll c$. Take $\psi\in C^2(\RR)$ so that $c\leq \psi'(t)$ and $|\psi''(t)|\leq \epsilon \psi'(t)$ for all $t\in\RR$. We call $\psi'(t)$ the {\it instantaneous sampling rate} (ISR) and $\psi'(t)/2$ the {\it instantaneous Nyquist frequency} (INF).
\end{defn}

Note that ISR and INF naturally generalize the notion of sampling rate and Nyquist frequency -- the higher the ISR is at a moment, the higher the sampling rate is around this moment. 
The condition $|\psi''(t)|\leq \epsilon \psi'(t)$ says that locally the sampling is close to a uniform sampling. 
A natural problem regarding the above definition of ISR and INF is the identifiability issue. Precisely, given a sampling points $\{t_m\}_{m\in\ZZ}$, we could find many different functions $\psi$ which leads to the same sampling points. We claim that under the provided condition of ISR and INF, they are well-defined up to an error of order $\epsilon$.

\begin{thm}
Fix $c>0$. If $\psi,\tilde{\psi}\in C^2(\RR)$ both satisfy the conditions for ISR and generate the same sampling points $\{t_m\}_{m\in\ZZ}$, then we have $|\psi'(t)-\tilde{\psi}'(t)|\leq 2\epsilon$ and $|\psi(t)-\tilde{\psi}(t)|\leq 2\epsilon(t_{m+1}-t_m)$ for all time $t\in[t_m,t_{m+1}]$. Globally, we have $|\psi'(t)-\tilde{\psi}'(t)|\leq 2\epsilon$ and $|\psi(t)-\tilde{\psi}(t)|\leq 2\epsilon/c$ for all time $t\in\RR$.
\end{thm}
\begin{proof}
Note that we have $\psi(t_m)=\tilde{\psi}(t_m)=t$ for all $m\in\ZZ$. Consider $t\in[t_m,t_{m+1})$. By a direct calculation, we have 
\begin{align}
|\psi'(t)-\tilde{\psi}'(t)|&\leq \int_{t_m}^t|\psi''(s)-\tilde{\psi}''(s)|d s\\
&\leq \epsilon \int_{t_m}^t(\psi'(s)+\tilde{\psi}'(s))ds\nonumber\\
&=\epsilon[(\psi(t)-\psi(t_m))+(\tilde{\psi}(t)-\tilde{\psi}(t_m))]\nonumber\\
&\leq 2\epsilon(\psi(t_{m+1})-\psi(t_m))=2\epsilon\nonumber,
\end{align}
where the second inequality holds by the assumption that $|\psi''(t)|\leq \epsilon \psi'(t)$ and $|\tilde{\psi}''(t)|\leq \epsilon \tilde{\psi}'(t)$ both hold, and the last inequality holds by the monotonic assumption of $\psi$ and $\tilde{\psi}$. To finish the proof, note that by the mean value theorem, we have $\psi(t_{m+1})-\psi(t_m)=\psi'(t')(t_{m+1}-t_m)$ for some $t'\in[t_m,t_{m+1}]$, which leads to 
\begin{equation}
t_{m+1}-t_m=\frac{\psi(t_{m+1})-\psi(t_m)}{\psi'(t')}=\frac{1}{\psi'(t')}\leq \frac{1}{c},
\end{equation}
where the last inequality holds by the assumption that $\psi'(t)\geq c$ and $\tilde{\psi}'(t)\geq c$.
Thus, we have
\begin{align}
|\psi(t)-\tilde{\psi}(t)|&\leq \int_{t_m}^t|\psi'(s)-\tilde{\psi}'(s)|d s\\
&\leq 2\epsilon (t-t_m)\leq 2\epsilon (t_{m+1}-t_m)\leq 2\epsilon/c\nonumber
\end{align}
and hence the proof is done.
\end{proof}
This theorem essentially says that the ISR for a given sampling points $\{t_m\}_{m\in\ZZ}$ is well defined up to the error of order $\epsilon$. We mention that for a randomly given sampling points $\{t_m\}_{m\in\ZZ}$ so that $t_{m+1}>t_m$, we may not be able to define a ISR function which satisfies the condition in Definition \ref{Definition:ISR}. In this paper, we focus on sampling scheme which satisfies Definition \ref{Definition:ISR}. 

In addition to ISR and INF, note that we could also naively generalize the notion of Nyquist rate to its instantaneous version.
\begin{defn} 
Take a signal $f(t)=a(t)\cos(2\pi\phi(t))\in\mathcal{A}^{c_1,c_2}_{\epsilon}$. We call $2\phi'(t)$ the {\it instantaneous Nyquist rate} (INR) of the function $f(t)$.
\end{defn}

Note that the INR reduces to the notion of Nyquist rate when $\phi(t)$ is linear and $a(t)$ is constant. We will always assume that $\psi'(t)>2\phi'(t)$; that is, at each time instant, we have at least two sampling points from an oscillation, otherwise the oscillatory information might be lost. Thus, this is again a natural generalization of the sampling theory under the uniform sampling scheme and the band-limited assumption.

\section{The reflection effect}\label{Section:reflection}

With the adaptive harmonic model and the samples, we would like to study the underlying dynamical features, like the IF and AM of the signal. A common practice to convert the discretized sampling $\mathcal{X}$ to a continuous function is via an interpolation scheme. 
We now show the structured artifacts caused by the commonly applied spline interpolation scheme. In particular, the IF of an artificial component in the interpolated signal is a reflection of $\phi'$ associated with the INF.  Further, while the INF is closer to the IF, this artifact is more severe, which is the case we encounter when we study the IHR or EDR signals.

\begin{thm}\label{claim:main}
Take $f(t)=a(t)\cos(2\pi\phi(t))\in\mathcal{A}^{c_1,c_2}_{\epsilon}$ and sample $f(t)$ with the ISR $\psi'(t)>2\phi'(t)$. Fix a $n$-th order spline interpolation, where $n\geq 1$, and denote the interpolated signal as $\tilde{f}_{n}(t)$. Then, we have
\begin{align}
\tilde{f}_{n}(t)= a(t) \sum_{k\in\ZZ} c_n(k)\cos[2\pi (k\psi(t)-\phi(t))]+O(\sqrt{\epsilon}),
\end{align}
where $c_{n}(k)\in\RR$ depends on $n$.
\end{thm}

\begin{proof}

The proof is divided into three steps. First, we show the proof when the signal is harmonic and the sampling is uniform; second, we show the result when the signal is in the adaptive harmonic model and the sampling is uniform; third, we show the general proof.

We start from the harmonic signal $f(t)=\cos(2\pi \alpha t)$, where $0<\alpha<1/2$. 
Without loss of generality, we assume that the sampling rate is $1$ Hz; that is, $\psi(t)=t$ and the ISR is $1$. If not, for example, $\psi(t)=kt$, where $k>1$, we could upwrap the time axis by $s=kt$, and get $\psi(s)=s$ and $f(s)=\cos(2\pi  (\alpha/k)s)$, and the argument holds.  
Note that $f$ is a band-limited harmonic function which is also in $\mathcal{A}_\epsilon^{c_1,c_2}$. Note that the IF of $f(t)$ is $\alpha$, which is less than $\psi'(t)/2$. The uniform sampling scheme corresponds to the weighted Dirac train, which is a tempered distribution $f_{\delta}:=\sum_{l\in\ZZ}f(l)\delta_{l}$, where $\delta_l$ is the Dirac delta measure supported at $l\in\ZZ$.  The $n$-th order spline interpolation scheme, where $n\in\NN$, is performed as a convolution of $f_{\delta}$ with the $n$-th order {\em fundamental cardinal spline function} $\eta_{(n)}$, which satisfies
\begin{equation}
\widehat{\eta_{(n)}}(\xi)=\left[\sum_{l\in\ZZ}\left(\frac{\sin(\pi(\xi-l))}{\pi(\xi-l)}\right)^{n+1}\right]^{-1}\left(\frac{\sin(\pi\xi)}{\pi\xi}\right)^{n+1},
\end{equation}
where $\widehat{\eta_{(n)}}$ means the Fourier transform of $\eta_{(n)}$.
{See \cite[(2.14)]{Aldroubi_Unser_Eden:1992} or \cite[(4.6.9)]{Chui:1992} for example.}
Precisely, the interpolated signal based on the $n$-th order spline interpolation, denoted as $\tilde{f}$, satisfies $\tilde{f}=\eta_{(n)}\star f_\delta$, where $\star$ is the convolution. Note that it is an interpolation and we have $\tilde{f}(l)=f(l)$ for all $l\in\ZZ$. By a direct calculation, we have $\widehat{f_{\delta}}(\xi)=\sum_{k\in\ZZ} \hat{f}(k+\xi)=\frac{1}{2}\sum_{k\in \ZZ}[\delta_{k+\alpha}+\delta_{k-\alpha}]$, which leads to $\widehat{\tilde{f}}(\xi)=\frac{1}{2}\sum_{k\in\ZZ} \widehat{\eta_{(n)}}(\xi)[\delta_{k+\alpha}+\delta_{k-\alpha}]$. As a result, we have
\begin{align}
\tilde{f}(t)= \sum_{k\in\ZZ} \widehat{\eta_{(n)}}(k-\alpha)\cos(2\pi (k-\alpha)t)\label{reflection:formula:harmonic}
\end{align}
We thus finish the proof when the signal is harmonic and the sampling is uniform. Indeed, the main reflection component associated with INF, $1/2$, is $\widehat{\eta_{(n)}}(1-\alpha)\cos(2\pi (1-\alpha)t)$.

Second, without loss of generality, we consider a general function $f(t)=a(t)\cos(2\pi\phi(t))\in\mathcal{A}^{c_1,c_2}_{\epsilon}$, where $\phi'(t)<1/2$. Again, we consider a uniform sampling scheme at the sampling rate 1 Hz; that is, $\psi(t)=t$. Take $l_0\in\ZZ$. Consider a local harmonic approximation of $f$ around $l_0$, denoted as $f^{(l_0)}(t)=a(l_0)\cos[2\pi(\phi(l_0)-l_0\phi'(l_0)+ \phi'(l_0)t)]$. By the adaptive harmonic model and the same argument as that in \cite{Daubechies_Lu_Wu:2011}, we have the control between $f$ and $f^{(l_0)}$ for all $s\in\RR$ 
\begin{equation}\label{Proof:TaylorExpansionf}
f(l_0+s)=f^{(l_0)}(l_0+s)+C|s|\epsilon,
\end{equation}
where $C$ depends on $c_2$ in the definition of $\mathcal{A}_{\epsilon}^{c_1,c_2}$ and is uniformly bounded for all $l_0$. Also recall the fact that the $n$-th order cardinal spline decays exponentially \cite[(4.6.2)]{Chui:1992}. 
By taking the facts of exponential decay of $\eta_{(n)}$ and the Taylor expansion (\ref{Proof:TaylorExpansionf}), we know that for $t\in[l_0-1/2,l_0+1/2]$, 
\begin{equation}
\left|\sum_{l\in\ZZ}(f(l)-f^{(l_0)}(l))\eta_{(n)}(t-l)\right|\leq 2C\epsilon\sum_{k\in\NN} k e^{-ck}=C'\epsilon
\end{equation}
for some constant $c,C'>0$.
Hence, by (\ref{reflection:formula:harmonic}) we have for all $t\in [l_0-1/2,l_0+1/2)$
\begin{align}
&\tilde{f}(t)=\,(\eta_{(n)}\star f_\delta)(t)=(\eta_{(n)}\star f^{(l_0)}_\delta)(t)+O(\epsilon)\nonumber\\
=&\,\sum_{k\in\ZZ} \widehat{\eta_{(n)}}(k-\phi'(l_0))a(t)\cos[2\pi (kt-\phi(t))]+O(\epsilon)\label{reflection:formula:adaptiveharmonic},
\end{align}
where we use the fact that $a(t_0)\cos[2\pi (\phi(l_0)-l_0\phi'(l_0)+(k+\phi'(l_0))t)]=a(t)\cos[2\pi (\phi(t)+kt)]+O(\epsilon)$.
As a result, we have the proof when $f\in\mathcal{A}^{c_1,c_2}_{\epsilon}$ and the sampling is uniform. Indeed, the main reflection component associated with INF, $1/2$, is $\widehat{\eta_{(n)}}(1-\phi'(l_0))a(t)\cos[2\pi (t-\phi(t))]$, which has the IF $1-\phi'(t)$.

Third, we consider the case when the sampling is non-uniform and the signal satisfies the adaptive harmonic model. We consider $f(t)=a(t)\cos(2\pi\phi(t))\in\mathcal{A}^{c_1,c_2}_{\epsilon}$ and the ISR $\psi$ so that $\psi'(t)>2\phi'(t)$ for all $t$. Denote the sampling points as $\mathcal{T}:=\{t_i\}_{i\in\ZZ}$, where $\psi(t_i)=i$. Without loss of generality, we focus on $\tau$ so that $\phi'(\tau)<1/2$ and $\psi'(\tau)=1>2\phi'(\tau)$; the other cases could be proved by scaling. To simplify the notation, we assume that $\tau=N/2$, where $N=2\lceil1/\sqrt{\epsilon}\rceil$. By the assumption of $\psi$, we know that over the interval $[0,N]$, $|\psi'(s)-1|\leq \sqrt{\epsilon}$ for all $s\in [0,N]$. Denote $I=\{i\in\ZZ|\,0\leq t_i\leq N\}\subset\ZZ$ and $J=\{0,1,\ldots,N\}\subset\ZZ$. Note that $|t_l-l|\leq C'|l-\tau|\epsilon$ for all $l\in I$, where $C'$ is a constant depending on $\psi$, so $|t_l-l|$ is bounded by $\sqrt{\epsilon}$ and we know that $|I|=|J|=N+1$. Denote $N_{n,j}$ to be the B-spline of order $n$ defined on $\mathcal{T}$ \cite[(10.2.32)]{DeVilliers:2012}; that is,
\begin{equation}
N_{n,j}(x):=(t_{j+n+1}-t_j)\sum_{k=j}^{j+n+1}\frac{(x-t_k)_+^n}{\Pi_{k\neq l=j}^{j+n+1}(t_l-t_k)},
\end{equation}
where $x_+^n$ is the truncated power defined as $x^n_+=x^n$ when $x\geq 0$ and $x^n_+=0$ when $x<0$.
Note that when the sampling is uniform, that is, when $t_j=j$, then $N_{n,j}(x)=N_n(x-j)$, where $N_n$ is the cardinal B-spline of order $n$ \cite[Theorem 10.2.7]{DeVilliers:2012} satisfying
\begin{equation}
N_n(x):=\frac{1}{n!}\sum_{k=0}^{n+1}(-1)^k\left(\hspace{-5pt}\begin{array}{c}m+1\\k\end{array}\hspace{-5pt}\right)(x-k)_+^n.
\end{equation}
To prove the theorem, note that the technique of using the Fourier transform of the fundamental cardinal spline function does not work here, due to the non-uniform sampling. However, due to the condition of $\psi$, we could control the difference between the non-uniform sampling case we discuss here and the uniform sampling case. Recall that the $n$-th order spline interpolation for the non-uniform sampling is finding the unique \cite[Theorem 10.2.2(d)]{DeVilliers:2012} $n$-th order spline $\bar{f}$ such that 
\begin{equation}
\bar{f}(x)=\sum_{j=-n}^{N-n}c_jN_{n,j}(x)
\end{equation}
where $\{c_{-n},c_{-n+1},\ldots,c_{N-n-1},c_{N-n}\}\subset \RR$ satisfies the conditions
\begin{equation}
\sum_{j=-n}^{N-n}c_jN_{n,j}(t_i)=f(t_i)
\end{equation}
for all $i=0,\ldots,N$ \cite[(10.3.7)]{DeVilliers:2012}. By the Schoenberg-Whitney Theorem \cite[Theorem 10.3.2]{DeVilliers:2012}, the interpolation is solved by finding $c_j$ by
\begin{equation}
\mathbf{c}=\mathbf{A}^{-1}\mathbf{f},
\end{equation}
where $\mathbf{A}$ is a $(N+1)\times (N+1)$ matrix 
\[
\mathbf{A}=\begin{bmatrix}N_{n,-n}(t_0) & N_{n,-n+1}(t_0) & \ldots & N_{n,N-n}(t_0)\\ N_{n,-n}(t_1) & N_{n,-n+1}(t_1) & \ldots & N_{n,N-n}(t_1) \\ \vdots & \vdots & \vdots & \vdots \\ N_{n,-n}(t_N) & N_{n,-n+1}(t_N) & \ldots & N_{n,N-n}(t_N) \end{bmatrix},
\]
$\mathbf{c}=[c_{-n},c_{-n+1},\ldots,c_{N-n-1},c_{N-n}]^T\in\RR^{N+1}$ and $\mathbf{f}=[f(t_0), f(t_1),\ldots,f(t_{N-n-1}),f(t_{N-n})]^T\in\RR^{N+1}$.
Similarly, denote the $n$-th order spline interpolation from the uniform sample $\{(i,f(i))\}_{i\in\ZZ}$ to be $\tilde{f}$, which satisfies
\begin{equation}
\tilde{f}(x)=\sum_{j=-n}^{N-n}d_jN_{n}(x-j)
\end{equation}
where $\{d_{-m},d_{-m+1},\ldots,d_{n-m-1},d_{n-m}\}\subset \RR$ satisfies the conditions
\begin{equation}
\sum_{j=-n}^{N-n}d_jN_{n}(i-j)=f(i)
\end{equation}
for all $i=0,\ldots,N$. The $d_j$ is again solved by applying the Schoenberg-Whitney Theorem. 
By the assumption of $\psi$, we know that $|f(t_l)-f(l)|\leq C|l-\tau|\epsilon$ for all $l \in I$ by (\ref{Proof:TaylorExpansionf}), where $C'$ depends on $\|f'\|_{\infty}$. Also, we have $|N_{n,j}(x)-N_n(x-j)|\leq C''|j-\tau|\epsilon=O(\sqrt{\epsilon})$ for all $j\in I$, where $C''$ is another constant depending on $n$. Hence, $c_j=d_j+O(\sqrt{\epsilon})$ by the $\sqrt{\epsilon}$-perturbation of $\mathbf{A}$. Also note that $N_{n,j}(x)$ and $N_n(x)$ are all compactly supported. As a result, over $[0,N]$, the $n$-th order spline interpolation over the non-uniform sampling $\mathcal{T}$ satisfies
\begin{align}
\bar{f}(x)&=\sum_{j=-n}^{N-n}c_jN_{n,j}(x)=\sum_{j=-n}^{N-n}d_jN_{n}(x-j)+O(\sqrt{\epsilon})\nonumber\\
&=\tilde{f}(x)+O(\sqrt{\epsilon}).
\end{align} 
Thus, by the above argument, we know that $\tilde{f}(x)$ satisfies the reflection property (\ref{reflection:formula:adaptiveharmonic}), and hence we obtain the proof.

\end{proof}

\section{Reflection effect and time-frequency analysis}\label{Section:Numerical}

In practice, the reflection effect discussed in Section \ref{Section:reflection} is pronounced when we apply the TF analysis techniques to study the time-varying dynamics hidden inside the signal, and the effect is even worsened when we apply the sharpening technique in the TF analysis. In this section and the next, we demonstrate how the interpolation induced reflection effect is pronounced by two TF analysis techniques, the RM \cite{Kodera_Gendrin_Villedary:1978,Auger_Flandrin:1995,Chassande-Mottin_Auger_Flandrin:2003} or the SST \cite{Daubechies_Lu_Wu:2011,Chen_Cheng_Wu:2014}. Similar phenomena could be found in other TF analysis methods, ranging from the linear to quadratic methods \cite{Flandrin:1999}. Note that the SST is a variation of the RM, and these techniques could be carried out to sharpen the TF representation determined by a chosen linear TF analysis, for example, the short time Fourier transform (STFT) or continuous wavelet transform (CWT). In a nutshell, RM and SST are nonlinear TF techniques aiming to alleviate the spreading effect in the TF representation determined by the linear TF analysis, which is caused by the inevitable Heisenberg uncertainty principle. These techniques sharpen the TF representation by taking the phase information hidden inside the linear TF analysis into account. For details, we refer the reader with interest to \cite{Auger_Flandrin_Lin_McLaughlin_Meignen_Oberlin_Wu:2013,Daubechies_Wang_Wu:2015} for the review material.

\begin{figure*}[tphb]
\centering
\includegraphics[width=0.45\textwidth]{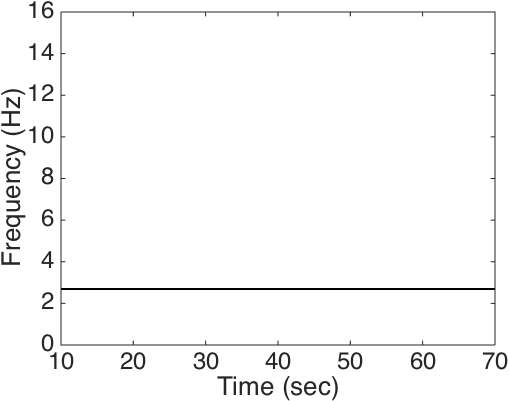}
\includegraphics[width=0.45\textwidth]{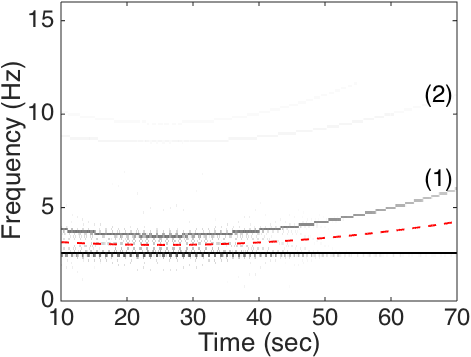}\\
\includegraphics[width=0.45\textwidth]{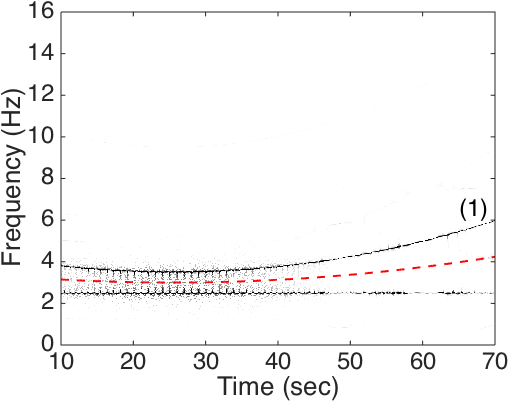}
\includegraphics[width=0.45\textwidth]{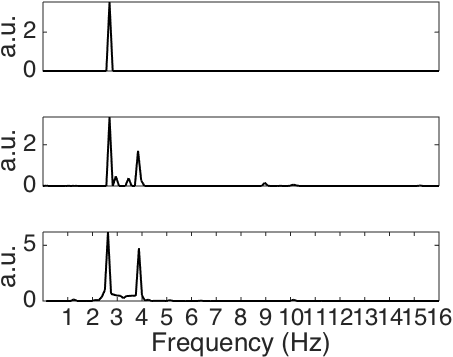}
\caption{The time-frequency (TF) representation of $f_1(t)=\cos(2\pi2.5t)$ and different interpolations from the sampling scheme $\psi_1$. The instantaneous frequency (IF) of $f_1$ is $\phi'=2.5$. 
Left: the SST-STFT result of $f_1(t)$. 
Middle left: the SST-STFT result of the cubic spline (CS) interpolation. The instantaneous Nyquist frequency (INF), $\psi_1'/2$, is superimposed as a red dashed curve on it. Here we could see components with IF $\psi_1'-\phi'$ marked as (1) and $\psi_1'+\phi'$ marked as (2). 
Middle right: the reassigned STFT result of the CS interpolation, where we could only see an extra component with IF $\psi_1'-\phi'$ marked as (1). Note that the artificial component with IF $\psi_1'-\phi'$ is the reflection of $\phi'$ associated with the INF. Right: from top to bottom we show the $40$-th second slice of the TF representations shown in left, middle left, middle and middle right. The unit of the y-axis is arbitrary (a.u.).}\label{Fig1}
\end{figure*}

To demonstrate this reflection effect, we start from a harmonic function $f_1(t)=\cos(2\pi 2.5t)$ and sample it with the ISR $\psi'_1(t)=6+\frac{(t-80/\pi)^2}{800}$. Note that obviously the ISR is greater than the INR of $f$. Then, we demonstrate the reflection effect by applying different common interpolation schemes. We then run SST-STFT and reassigned STFT on the interpolated signal, where the interpolated signal is sampled uniformly at $64$ Hz. To avoid possible boundary effects, we sample the signal for $80$ seconds. The results are shown in Figure \ref{Fig1}. In this study, {the TF representation $R\in\CC^{n\times m}$ is displayed in the log scale. Precisely, we plot $\tilde{R}\in\RR^{n\times m}$, where $\tilde{R}_{i,j}=\max\{10^{-2},\log(1+\min\{|R_{i,j}|,q)\} ) \}$, where $q$ is the $99.8\%$ quantile of all entries of $|R|$.}

We next consider the same procedure on a non-harmonic function $f_2(t)=(0.7+t^{1.1})\cos(2\pi(\pi t+0.2\cos(t)))$, whose AM and IF are $a(t)=0.7+t^{1.1}$ and $\phi'(t)=\pi-0.2\sin(t)$, respectively, and we take another ISR, $\psi'_2(t)=8+0.5\cos(\pi t/10)$. Note that $\psi'_2$ is greater than INR of $f_2$. The result is shown in Figure \ref{Fig2}. Note that we could see a clear reflected component associated with the INF in all the above cases, and the behavior depends on the setup.

\begin{figure*}[tphb]
\centering
\includegraphics[width=0.45\textwidth]{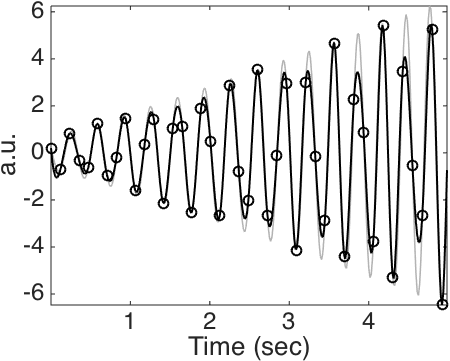}
\includegraphics[width=0.45\textwidth]{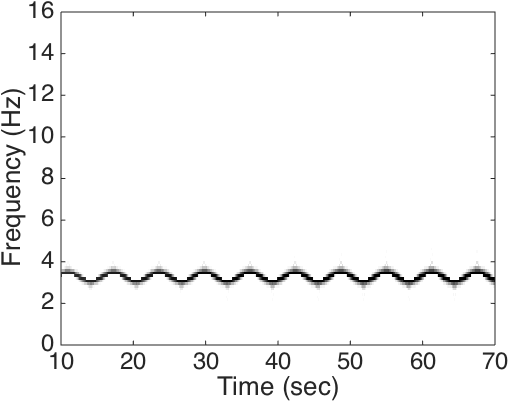}\\
\includegraphics[width=0.45\textwidth]{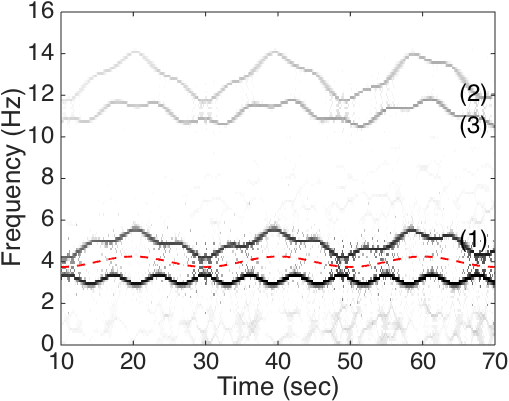}
\includegraphics[width=0.45\textwidth]{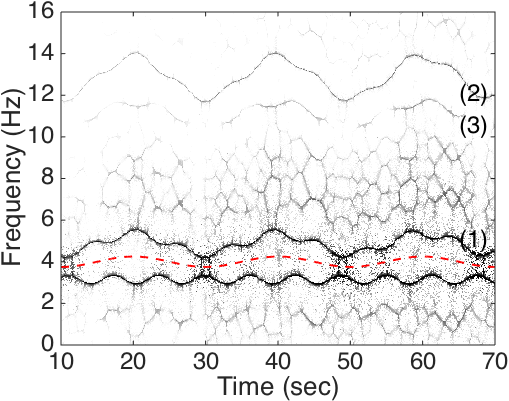}
\caption{The time-frequency representation of $f_2(t)=(0.7+t^{1.1})\cos(2\pi(\pi t+0.2\cos(t)))$ and the interpolation from the sampling dataset associated with $\psi_2(t)$. The instantaneous frequency (IF) of $f_2$ is $\phi'(t)=\pi-0.2\sin(t)$.
Left: $f_2(t)$ is shown as a gray curve with the non-uniform samples superimposed as black circles. The cubic spline (CS) interpolation is shown as the black curve. The unit of the y-axis is arbitrary (a.u.).
Middle left: the SST-STFT result of $f_2(t)$.
Middle right: the SST-STFT result of the CS interpolation. 
Right: the reassigned STFT result of the CS interpolation. 
Note that in the middle, middle right and right subfigures, in addition to $f_2(t)$, we could see components with IF $\psi_2'-\phi'$ marked as (1), $2\psi_2'-\phi'$  marked as (2) and $\psi_2'+\phi'$ marked as (3). Note that the artificial component with IF $\psi_2'-\phi'$ is the reflection of $\phi'$ associated with the INF.
}\label{Fig2}
\end{figure*}

\section{Real signal from Anesthesia}\label{Section:Real}

General anesthesia is usually inevitable for a patient receiving major surgery. For the short term and long term well-being of the patient, the anesthetic agent dose should be dynamically adjusted to achieve an adequate level of anesthesia. It has been shown that the oscillatory patterns in the R-to-R peak interval (RRI) time series of electrocardiography and the respiration contain a lot of information regarding anesthesia dynamics \cite{Huang_Chan_Lin_Wu_Huang:1997,Lin_Hseu_Yien_Tsao:2011,Lin_Wu_Tsao_Yien_Hseu:2014,Wu_Chan_Lin_Yeh:2014,Lin:2015Thesis,Chui_Lin_Wu:2014}. %
It has been shown in \cite{Lin_Hseu_Yien_Tsao:2011,Lin_Wu_Tsao_Yien_Hseu:2014,Wu_Chan_Lin_Yeh:2014,Lin:2015Thesis} that the adaptive harmonic model, the multi-taper reassigned STFT \cite{Xiao_Flandrin:2007,Babadi_Brown:2014} and the multi-taper SST-STFT \cite{Lin_Wu_Tsao_Yien_Hseu:2014,Lin:2015Thesis} serve as a good framework toward the goal, especially when noise is inevitable. Precisely, the oscillatory behavior of the RRI during anesthesia could be modeled by the adaptive harmonic model, and how strong the component is related directly to the anesthetic depth. 
We would extract stable features by the multi-taper reassigned STFT and the multi-taper SST-STFT to better quantify the anesthetic depth.
While its clinical value has been shown in different problems, its sampling theory issue is left unanswered to the best of our knowledge. 
We mention that the sampling issue for the power spectrum approach to study HRV based on the stationarity assumption has been widely studied; see for example \cite{Laguna_Moody_Mark:1998,Singh_Vinod_Saxena:2004}. 
It could be argued that reassigned STFT and SST-STFT add complications on the top of a spectrum approach, but note the difference between the underlying models which are aiming to capture different phenomena.

Here we demonstrate two examples regarding this direction which might generate potential artifact in the TF representation -- IHR analysis and EDR analysis. 
The IHR and EDR are acquired from the recorded ECG signal in the following way. Denote the recorded lead II ECG signal as $E(t)$ which is digitalized at the sampling rate $1000$Hz. The R peak detection was automatically determined from $E(t)$. The ECG signal is clean without significant noise contamination, and no ectopic beats nor electro-cauterization happen in the recorded signal. The collected RRI time series is denoted as $\mathcal{X}=\{t_i,t_{i+1}-t_i\}_{i=1}^N$, where $t_i\in\RR$ is the time stamp of the $i$-th R peak. Then, we follow the common practice and approximate the IHR from $\mathcal{X}$ by the cubic spline interpolation \cite{TaskForce:1996}, and denote the approximated IHR as $\tilde{r}_{\textsf{m}}$.  
Next, $\tilde{r}_{\textsf{m}}$ is resampled to be equally spaced at $8$ Hz for the multi-taper SST-STFT and multi-taper reassigned STFT analysis \cite{Singh_Vinod_Saxena:2004}, as it is commonly believed that most useful information inside IHR is below $0.5$ Hz under the stationary assumption.
{In addition to the time stamps of the R peaks,} we also have a non-uniform sampling dataset $\{t_i,E(t_i)\}_{i=1}^N$, where $E(t_i)$ is the amplitude of the $i$-th R peak. The EDR signal, a surrogate of the respiratory signal denoted as $\tilde{R}(t)$, is built up by applying the cubic spline interpolation on $\{t_i,E(t_i)\}_{i=1}^N$. We mention that although the amplitude scales of $R(t)$ and $\tilde{R}(t)$ are different, they share the same oscillatory information inside the respiratory signal, in particular the IF. We refer readers with interest in EDR to \cite{Moody_Mark_Zocoola_Mantero:1985,Chui_Lin_Wu:2014} for details.
To confirm the existence of the reflection effect as an artifact in the EDR signal, when we record the ECG signal, we simultaneously record the airway flow signal, which is denoted as $R(t)$. To avoid any possible artifacts, the airway flow signal is uniformly sampled at the sampling rate $25$ Hz. As R peaks are viewed as the non-uniform sampling of the IHR and EDR, the ISR, denoted as $\psi'(t)$, could be estimated by the cubic spline interpolated function from $\left\{t_i,(t_{i+1}-t_i)^{-1}\right\}_{i=1}^N$. 

{As we have the airflow signal serving as the ground truth for the respiratory signal, we start from showing the result of the EDR signal.} The results of multi-taper SST-STFT and the multi-taper reassigned STFT on $\tilde{R}(t)$ based on the cubic spline interpolation and the airway flow signal are shown in Figure \ref{Fig5}. 
Here the multiples of the base oscillatory component with the IF at about $0.5$Hz are associated with the notion called ``wave-shape function'', for which we refer reader with interest to \cite{Wu:2013}.

\begin{figure}[t]
\centering
\includegraphics[width=0.6\textwidth]{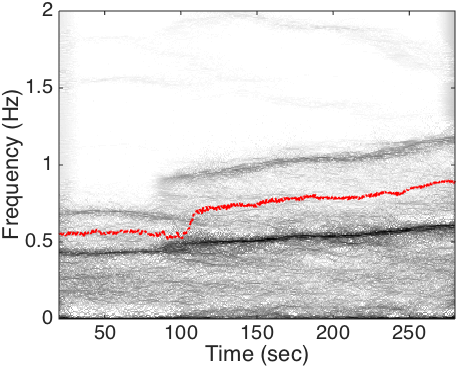}
\includegraphics[width=0.6\textwidth]{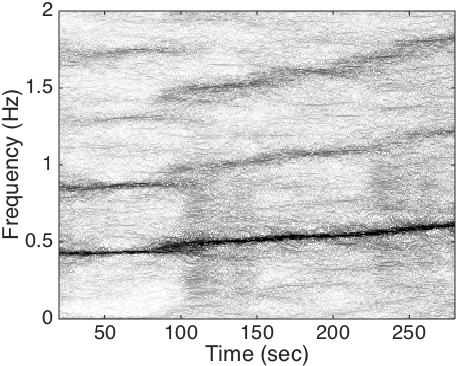}
\caption{The results of the ECG-derived respiration (EDR) signal and the airflow signal recorded simultaneously. The signals here are the same as those shown in Figure \ref{FigExample}, while we show here the whole $280$ second signal. The instantaneous Nyquist frequency (INF) is superimposed as a red dashed curve. Left: the multi-taper synchrosqueezed STFT (SST-STFT) result of the EDR based on the cubic spline (CS) interpolation. Right: the multi-taper SST-STFT of the airflow signal. It is clear that the base component with IF about $0.5$ Hz in the airflow signal is well captured by the EDR signal, while in the EDR signal there is an artificial reflected component associated with the INF. Also note the temporal reassignment effect in the multi-taper reassigned STFT around $100$-th second in in Figure \ref{FigExample} -- the TF resolution of the multi-taper reassignment STFT is sharper than that of the multi-taper SST-STFT shown, so the reflection effect is more visible.}\label{Fig5}
\end{figure}

The analysis results of the IHR signal based on multi-taper SST-STFT and multi-taper reassigned STFT are shown in Figures \ref{Fig40} and \ref{Fig4}. In Figure \ref{Fig40}, the whole 30 minutes analysis result is shown, while in Figure \ref{Fig4} we illustrate a zoom-in results of a $6$ minutes sub-interval. The TF representation provides several informations, in particular the time-varying frequency of the IHR signal. Note that there is a dominant curve near $0.5$ Hz, denoted as $x_1$, which is associated with the well-known phenomena of respiratory sinus arrhythmia (the respiratory signal is not shown), which oscillates at frequency about $0.5$ Hz. 
It is clear that we can see the reflected effect in the interpolated IHR $\tilde{r}_{\textsf{m}}$. Indeed, the dominant curve around $0.7$ Hz is the reflection of the component $x_1$ which is related to the INF $\psi'(t)/2$.

\begin{figure*}[tphb]
\centering
\includegraphics[width=0.95\textwidth]{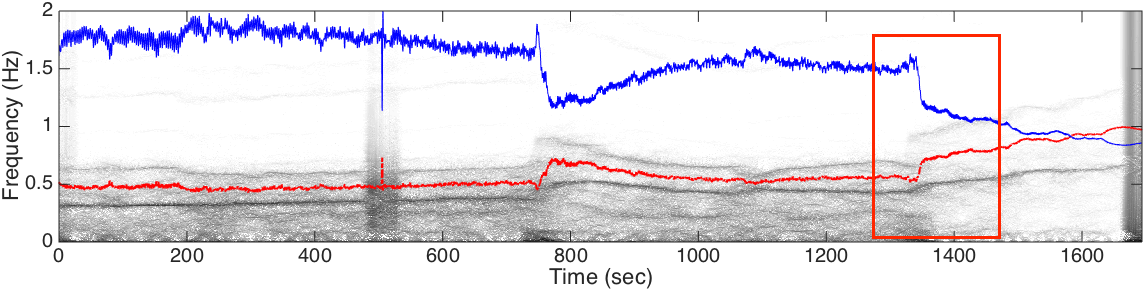}\\
\includegraphics[width=0.95\textwidth]{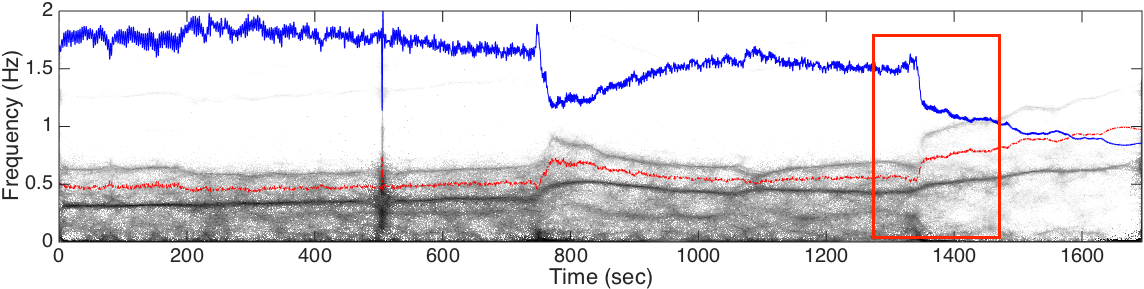}
\caption{The time-frequency representation of the IHR signal, $\tilde{r}_{\textsf{m}}(t)$, over the whole $30$ minutes period. The instantaneous Nyquist frequency is superimposed as a red dashed curve on it; the IHR signal is superimposed as a blue curve. Top: the multi-taper SST-STFT result of the IHR based on the cubic spline (CS) interpolation; bottom: the multi-taper reassigned-STFT result of the IHR based on the CS interpolation. The signal and the TF representation in the red box are zoomed in and displayed in Figure \ref{Fig4}.}\label{Fig40}
\end{figure*}

\begin{figure}[!t]
\centering
\includegraphics[width=0.95\textwidth]{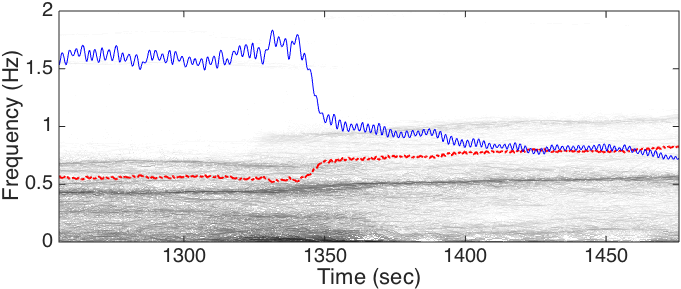}\\
\includegraphics[width=0.95\textwidth]{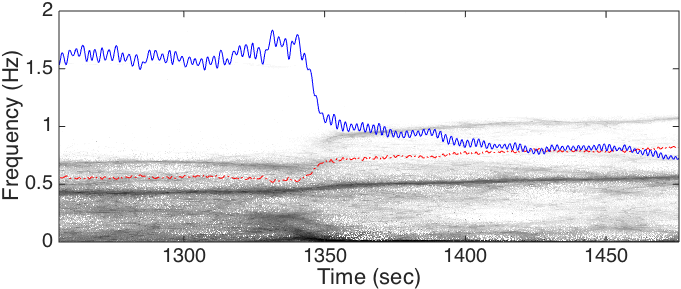}
\caption{The time-frequency (TF) representation of the IHR signal, $\tilde{r}_{\textsf{m}}(t)$, over a 6 minutes period. This is a zoom-in illustration of Figure \ref{Fig40} for the sake of showing the relationship between the IHR signal and the time-frequency representation. The instantaneous Nyquist frequency is superimposed as a red dashed curve on it; the IHR signal is superimposed as a blue curve. Top: the multi-taper SST-STFT result of the IHR based on the cubic spline (CS) interpolation; bottom: the multi-taper reassigned-STFT result of the IHR based on the CS interpolation. Note that the temporal reassignment in the multi-taper reassigned-STFT sharpen the reflected component around the period around the $1350$-th second. It is clear that the IHR oscillates faster after $1350$-sec, which leads to a higher instantaneous frequency. }\label{Fig4}
\end{figure}

To close this section, we demonstrate one more example with a different interpolation scheme, the shape-preserving piecewise cubic interpolation (PCHIP). The TF representation of the EDR signal generated via the PCHIP interpolation scheme determined by the multitaper SST is shown in Figure \ref{Fig8}. It is clear that in the PCHIP interpolation scheme, the reflection effect still exists. This demonstration shows that the reflection effect is not specialized to the spline interpolation. In Figure \ref{Fig8}, we also show the TF representation of the EDR signal generated via the PCHIP interpolation scheme determined by STFT. The STFT is carried out with the Gaussian window. It is also clear that we could see the reflection effect.

\begin{figure}[!t]
\centering
\includegraphics[width=0.7\textwidth]{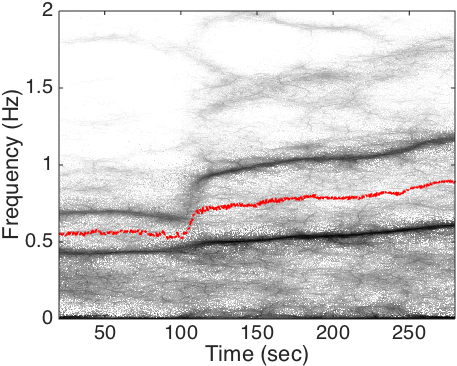}
\includegraphics[width=0.7\textwidth]{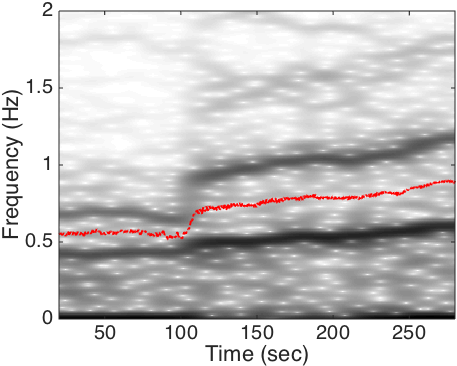}
\caption{The time-frequency representation of the ECG-derived respiration (EDR) signal generated via the piecewise cubic interpolation (PCHIP) interpolation scheme. The underlying ECG signal for the EDR signal is the same as that shown in Figure \ref{Fig5}. The instantaneous Nyquist frequency (INF) is superimposed as a red dashed curve. Left: the multi-taper reassigned short time Fourier transform (STFT) result of the EDR based on the PCHIP interpolation. Right: the STFT of the EDR based on the PCHIP interpolation. It is clear that the reflection effect exists. Also note that the multi-taper reassigned STFT is sharper than the STFT.}\label{Fig8}
\end{figure}

\section{Possible solutions to the reflection effect}\label{Section:PossibleSolution}

A naive solution to the reflection effect is by the hard threshold. Precisely, given a TF representation $R:\RR\times \RR^+\to \CC$, we could define a new TF representation $\tilde{R}$ by setting the TF representation coefficients with frequency larger than INF to zero; that is,
\begin{equation}
\tilde{R}(t,\xi)=\left\{
\begin{array}{ll}
R(t,\xi)& \mbox{ when }\xi\leq\psi'(t)/2\\
0& \mbox{ when }\xi>\psi'(t)/2
\end{array}
\right..
\end{equation}
This will directly remove the reflection artifact induced by the spline interpolation. Although such an adaptive filtering could mitigate the problem, due to the nonstationarity, this approach might also remove the possible information hidden inside the signal. Another naive solution is to take a higher order spline interpolation so that in the Fourier domain the spectrum of the fundamental cardinal spline is closer to the step function. Indeed, it is well known that when $n\to \infty$, $\tilde{\eta}^{(n)}$ converges to $\chi_{[-1,1]}$ in the $L^p$ sense \cite[(16)]{Unser:1999}. See Figure \ref{Fig7} for an example of the EDR signal generated by the $12$-th order spline interpolation. It is clear that compared with the cubic spline interpolation result shown in Figures \ref{FigExample} and \ref{Fig5}, the reflection effect is alleviated. While this approach seems to work, however, it is well known that the higher the order of spline interpolation is, the more severe the overfitting is. This fact might limit its application. Depending on the application, we could consider different penalty to determine the optimal order of spline interpolation. Yet another naive solution is to pre-process the signal by a low pass filter to the interpolated signal to reduce the reflection artifact. However, this approach only works when the signal is band-limited -- in general, for example the adaptive harmonic model, the structured artifact might not be removed but perturbed by the low pass filter, which leads to more complicated artifacts. 
These findings suggest that we should consider to design a different interpolation scheme to balance between the the reflection effect and the interpolation purpose. This opens a future research direction in the TF analysis.

\begin{figure}[!t]
\centering
\includegraphics[width=0.7\textwidth]{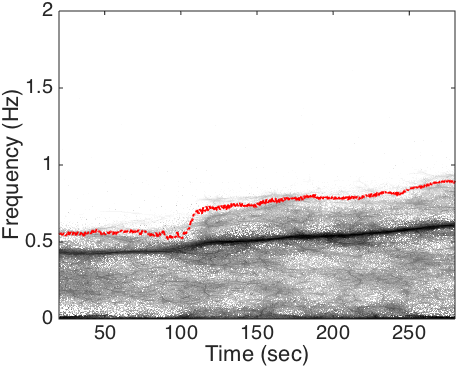}
\includegraphics[width=0.7\textwidth]{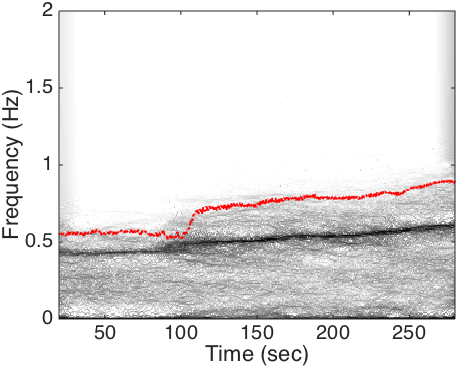}
\caption{The time-frequency representation of the ECG-derived respiratory (EDR) signal determined by the $12$-th order spline interpolation. The ECG signal we use to construct the EDR signal here is the same as that used for the EDR signal shown in Figure \ref{FigExample}. The instantaneous Nyquist frequency is superimposed as a red dashed curve on it. Left: the multi-taper reassigned-STFT result of the EDR signal determined by the $8$-th order spline interpolation; right: the multi-taper SST-STFT result of the EDR signal determined by the $8$-th order spline interpolation. Clearly, compared with Figures \ref{FigExample} and \ref{Fig5}, the reflection effect is reduced.}\label{Fig7}
\end{figure}

\section{Discussion}\label{Section:Conclusion}

While extracting oscillatory features from a given time series via TF analysis is attracting more interest in bio-medical field, in particular when people want to study the dynamics, the artifacts are enhanced, in particular the reflection effect we discuss in this paper. Thus, when apply the TF analysis to study the bio-medical signal, in addition to search for features inside the TF representation, it is important to pay attention not only to the sampling effect but also to the interpolation scheme, which might generate artifacts and introduce misleading interpretations. 

The evidence of the reflection effect is shown clearly in the simulated results in Section \ref{Section:Numerical}. In addition, it is clear that a reflection component exists in the EDR signal, as is shown in Figure \ref{Fig5} in Section \ref{Section:Real}, while the reflection component does not exist in the airflow signal. While the reflection component has a quite different IF behavior compared with the base component, if not being careful, these observations might lead to a misinterpretation that there is a possible hidden structure inside the respiratory signal. The same comment holds for the IHR signal. The reflective component shown in Figure \ref{Fig40} might lead to a misinterpretation that there are two oscillatory components inside the IHR; however, it is actually originated from the numerical interpretation. In order to get the correct information, few possible methods to alleviate the reflection effect is proposed in Section \ref{Section:PossibleSolution}.

In addition to designing a different interpolation scheme raised in Section \ref{Section:PossibleSolution}, there are several open problems and their importances have been indicated in this paper. First, while there are several different interpolation schemes available for reconstructing the signal from the (non-)uniform sampling points, depending on the application, it is important to pay attention to the interplay between the TF analysis and the artifact generated by the interpolation scheme, so that the results are not misled by these artifacts. 
Second, it raises a question specific to the HRV or BRV -- as the sampling scheme is limited by the physiological facts, it is not possible to sample the instantaneous heart rate as fast as possible. Thus, is it possible to evaluated the dynamical spectral information directly from the sampled time series, so that we could avoid the need of the interpolation?   
One possible direction is the following. For a given non-uniform sampled time series, the Lomb periodogram \cite{Laguna_Moody_Mark:1998} or the spectrum of counts \cite{DeBoer:1985} are commonly applied methods if we want to estimate the power spectrum directly from the non-uniform sampled time series. While they allow us to avoid the interpolation, however, the necessary phase information for the reassignment technique and SST is missed in these methods. Thus, it deserves a further study to explore the possibility to combine the benefits of the Lomb periodogram or the spectrum of counts and the reassignment technique.
\section{Conclusion}

We report a potential artifact shown in the interaction of the interpolation and the TF analysis. This artifact is theoretically studied under the spline interpolation scheme, and is referred to as the reflection effect. A solution as well as the future directions are provided. While in this paper we demonstrate the evidence based on the bio-medical signals, including the IHR signal and the EDR signal, this theoretical phenomena might happen in other fields. In conclusion, when we apply the TF analysis techniques to study the time-varying dynamics hidden inside a time series, it is important to pay attention to avoid any mis-interpretation, in particular when the sampling rate is not significantly higher than the spectrum we are interested in.

\section*{Acknowledgement}
Hau-tieng Wu would like to thank Professor Charles K. Chui, Professor Ingrid Daubechies and Professor Andrey Feuerverger for fruitful discussions. Hau-tieng Wu acknowledges the support of Sloan Research Fellow FR-2015-65363. Part of this work was done during Hau-tieng Wu's visit to National Center for Theoretical Sciences (NCTS), Taiwan; he would like to thank NCTS for its hospitality. The research of Yu-Ting Lin was supported by Shin Kong Wu Ho-Su Memorial Hospital (SKH-8302-102-DR-27).

\bibliographystyle{IEEEtran} 
\bibliography{MedicineVariability,TFanalysis}

\end{document}